\documentclass[12pt,onecolumn,romanappendices]{IEEEtran}

\usepackage{cite}
\usepackage{amsmath,amssymb,amsfonts,amscd,amsthm,amsbsy}
\usepackage{algorithmic}
\usepackage{graphicx}
\usepackage{textcomp}
\usepackage{wrapfig}
\usepackage{dsfont}
\usepackage[english]{babel}
\usepackage[ruled,vlined]{algorithm2e}
\usepackage{threeparttable}
\usepackage{textcomp}
\usepackage{multirow}
\usepackage{mathrsfs}
\usepackage{bbm}
\usepackage{mathtools}
\usepackage{array}
\usepackage{multirow}

\newtheorem{theorem}{Theorem}
\newtheorem{lemma}{Lemma}

\newtheorem{corollary}{Corollary}
\theoremstyle{definition}

\newtheorem{mydef}{Definition}
\theoremstyle{definition}

\theoremstyle{definition}
\newtheorem{example}{Example}

\theoremstyle{definition}
\newtheorem{const}{Code Construction}
\theoremstyle{definition}

\makeatletter
\renewcommand*\env@matrix[1][c]{\hskip -\arraycolsep
  \let\@ifnextchar\new@ifnextchar
  \array{*\c@MaxMatrixCols #1}}
\makeatother

\usepackage{etoolbox}
\AtEndEnvironment{example}{\null\hfill\IEEEQED}%

\newcommand{\Fb}{\mathds{F}}

\newcommand{\xb}{{\bf{x}}}
\newcommand{\eb}{{\bf{e}}}
\newcommand{\zb}{{\bf{z}}}
\newcommand{\yb}{{\bf{y}}}
\newcommand{\Hb}{{\bf{H}}}
\newcommand{\Pb}{{\bf{P}}}
\newcommand{\Ib}{{\bf{I}}}
\newcommand{\Cb}{{\bf{C}}}
\newcommand{\cb}{{\bf{c}}}
\newcommand{\ub}{{\bf{u}}}
\newcommand{\ab}{{\bf{a}}}
\newcommand{\ei}{\epsilon}
\newcommand{\csr}{\mathscr{C}}
\newcommand{\Ab}{{\textbf{A}}}

\newcommand{\Sb}{{\textbf{S}}}

\newcommand{\Kb}{{\textbf{K}}}
\newcommand{\Mb}{{\textbf{M}}}
\newcommand{\ical}{\mathcal{I}}
\newcommand{\bcal}{\mathcal{B}}
\newcommand{\acal}{\mathcal{A}}
\newcommand{\scal}{\mathcal{S}}

\newcommand{\wtm}{\mathrm{wt}}
\newcommand{\optm}{\mathrm{opt}}
\newcommand{\rcovm}{r_{\mathrm{cov}}}
\newcommand{\xcal}{\mathcal{X}}
\newcommand{\cd}{\mathrm{FIC}}
\newcommand{\thm}{\mathrm{th}}

\usepackage{psfrag}

\usepackage{paralist}

\usepackage{xfrac}

\title{Codes for Updating Linear Functions over\\ Small Fields}

\author{
\IEEEauthorblockN{Suman Ghosh and Lakshmi Natarajan}

\thanks{The authors are with the Department of Electrical Engineering, Indian Institute of Technology Hyderabad, Sangareddy 502\,285, India (email: \{ee16resch11006,\,lakshminatarajan\}@iith.ac.in).
}
}

\begin{document}

\maketitle

\begin{abstract}
We consider a point-to-point communication scenario where the receiver intends to maintain a specific linear function of a message vector over a finite field. When the value of the message vector changes, which is  modelled as a sparse update, the transmitter broadcasts a coded version of the modified message while the receiver uses this codeword and the current value of the linear function to update its contents. It is assumed that the transmitter has access to only the modified message and is unaware of the exact difference vector between the original and modified messages. 
Under the assumption that the difference vector is sparse and that its Hamming weight is at the most a known constant, the objective is to design a linear code with as small a codelength as possible that allows successful update of the linear function at the receiver. 
This problem is motivated by applications to distributed data storage systems.
Recently, Prakash and M\'{e}dard derived a lower bound on the codelength, which is independent of the size of the underlying finite field, and provided constructions that achieve this bound if the size of the finite field is sufficiently large.
However, this requirement on the field size can be prohibitive for even moderate values of the system parameters.
In this paper, we provide a field-size aware analysis of the function update problem, including a tighter lower bound on the codelength, and design codes that trade-off the codelength for a smaller field size requirement.
We also show that the problem of designing codes for updating linear functions is related to functional index coding or generalized index coding.
We first characterize the family of function update problems where linear coding can provide reduction in codelength compared to a naive transmission scheme.
We then provide field-size dependent bounds on the optimal codelength, and construct coding schemes based on error correcting codes and subspace codes when the receiver maintains linear functions of striped message vector.
These codes provide a trade-off between the codelength and the size of the operating finite field,
and whenever the achieved codelengths equal those reported by Prakash and M\'{e}dard the requirements on the size of the finite field are matched as well.
Finally, for any given function update problem, we construct an equivalent functional index coding or generalized index coding problem such that any linear coding scheme is valid for the function update problem if and only if it is valid for the constructed functional index coding problem.
\end{abstract}

\section{Introduction}
We consider a point-to-point communication scenario as shown in Fig.~\ref{fig:1} where the receiver maintains a linear function $\Ab\xb$ of a message vector $\xb$. 
The message $\xb$ is an $n$-length column vector over a finite field $\Fb_q$, where $q$ is any prime power, and $\Ab$ is an $m \times n$ matrix over $\Fb_q$ with $m \leq n$ and rank$(\Ab)=m$. 
Suppose the value of the message vector is updated to $\xb+\eb$, where $\eb$ represents a sparse update to the message, i.e., we assume that $\wtm(\eb) \leq \ei$ where $\wtm$ denotes the Hamming weight of a vector and $\ei$ is a known constant.
In other words at the most $\ei$ entries of the original message $\xb$ are updated to new values.
We assume that the transmitter has access to the updated message $\xb+\eb$, but is unaware of the original message $\xb$ or the sparse update $\eb$.
Note that the message update is modelled here as substitutions only and not as insertions or deletions.
The objective is to design a linear encoder that uses an $l \times n$ matrix $\Hb$ to generate the codeword $\cb=\Hb(\xb+\eb)$, with as small a codelength $l$ as possible, such that the receiver can decode $\Ab(\xb+\eb)$ using the transmitted codeword $\cb$ and the older version of its content $\Ab\xb$.

\begin{figure}[ht!]
\centering
\includegraphics[width=4.3in]{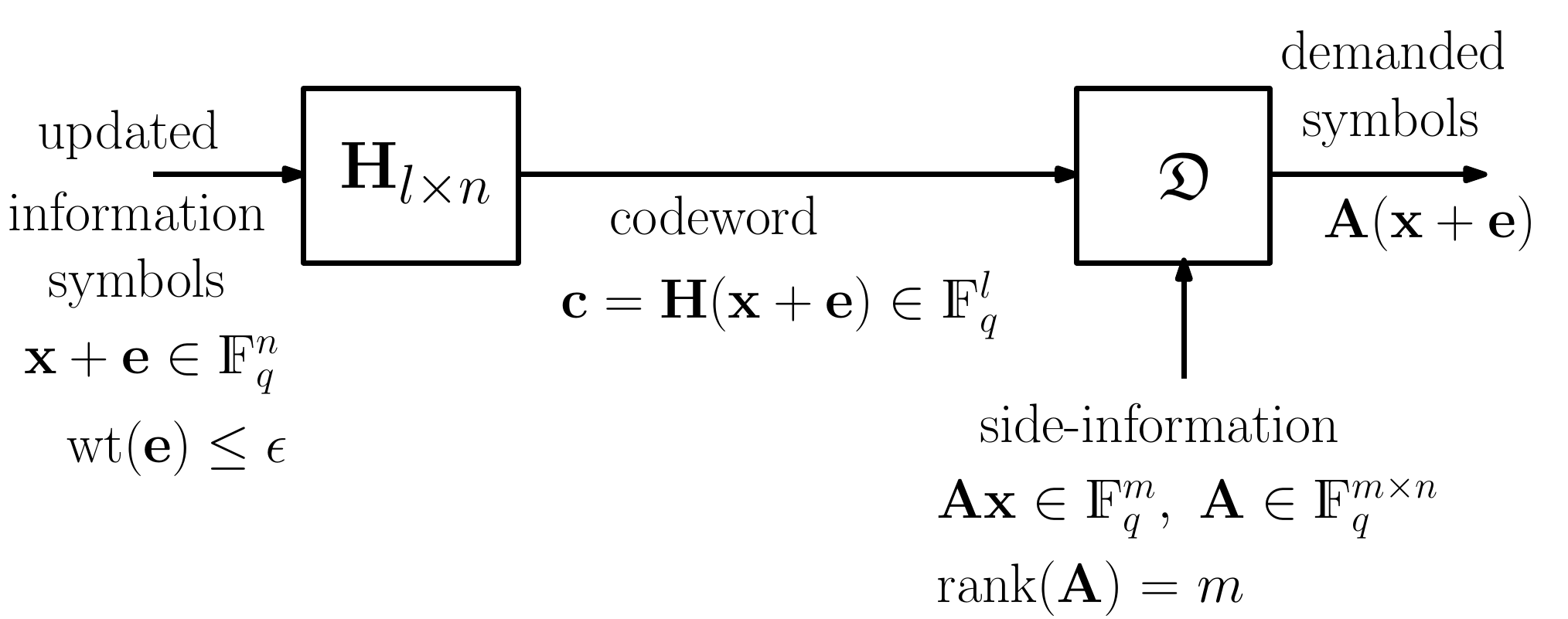}
\caption{System model for the point-point function update problem.}
\label{fig:1} 
\end{figure}

The problem is motivated by distributed storage systems (DSS) where information is stored in linearly coded form across a number of nodes to provide resilience against storage node failures~\cite{PM_2018}. In the scenario where multiple users can simultaneously edit a single file stored in a DSS, it is possible that a user who wishes to apply his update $\xb+\eb$ is unaware of the current version of the message $\xb$ stored in the DSS, for instance when another user has recently edited this file. 
Letting the user first learn the version $\xb$ stored in the DSS, and then apply his update will incur additional communication cost. 
As an alternative, if it is known that the update vector $\eb$ is sparse, it is possible to design schemes that do not require the knowledge of the value of $\eb$ at the transmitter~\cite{PM_2018,NSR_2014,NSRSR_update_2014}.

The function update problem was considered in~\cite{NSR_2014,NSRSR_update_2014} for DSS's for updating one of the storage nodes with the help of the other nodes in the system.
Note that each node in a DSS stores a linear function of the message.
A node can become \emph{stale} in such systems, for instance if the node goes offline while the message and the corresponding linear functions stored in the other nodes undergo an update. Once it is back online, the stale node connects to the other nodes in the distributed storage system to update its own linear function, and the stale data already stored in this node acts as side information. 
The authors of~\cite{NSR_2014,NSRSR_update_2014} design both the code for distributed storage and the code for function update to minimize the amount of data downloaded by the stale node to update its contents. This is unlike the problem statement considered in~\cite{PM_2018} as well as this paper, where it is  assumed that an arbitrary matrix $\Ab$ is given and a code for updating the function $\Ab\xb$ is to be designed.

The authors in~\cite{PM_2018} also consider a broadcast scenario where a codeword is broadcast to multiple nodes in order to update the (different) linear functions stored in each of the nodes.
Problems related to updating linear functions have been considered in~\cite{MPMY_2018,AC_2017,WC_2018}.
In~\cite{MPMY_2018}, codes for updating linear functions are used in cache-aided networks to reduce the cost of multicasting a sequence of correlated data frames. 
The problem of efficiently storing multiple versions of a file in a DSS while ensuring a property called \emph{consistency} is considered in~\cite{AC_2017,WC_2018}.



In the study of the point-to-point function update problem given in \cite{NSR_2014,NSRSR_update_2014,PM_2018} the authors derive the following field-size independent lower bound on the codelength 
\begin{equation*}
l \geq \min(m,2\ei). 
\end{equation*} 
Note that, if $m \leq 2\ei$, the lower bound on the codelength $l \geq m$ can be trivially achieved by transmitting $\Ab(\xb+\eb)$.
Hence, we will always assume that $m > 2\ei$.
The results in \cite{PM_2018} show that codelength $l=2\ei$ is achievable using maximally recoverable subcodes of $\csr_A$, the subspace spanned by the rows of $\Ab$, which are guaranteed to exist if the field size $q \geq 2\ei n^{2\ei}$. 
Note that this requirement imposed on the field size can be large even for moderate values of $\ei$ and $n$. 
The authors of~\cite{PM_2018} also consider the special case where the matrix $\Ab$ is \emph{striped}, i.e.,
\begin{equation*}
\Ab=\Ib_{a} \otimes \Cb = 
\begin{bmatrix}
\Cb & \pmb{0} & \dots & \pmb{0}\\
\pmb{0} & \Cb & \dots & \pmb{0}\\
\vdots & \vdots & \ddots & \vdots\\
\pmb{0} & \pmb{0} & \dots & \Cb
\end{bmatrix}
\end{equation*}
where $\Ib_{a}$ is the $a \times a$ identity matrix, $\mathbf{C} \in \Fb_q^{t \times K}$  and $\otimes$ denotes the Kronecker product. 
Note that $m=at$ and $n=aK$.
This structure frequently arises in distributed storage systems where the $n$-length data $\xb$ is partitioned into $a$ subvectors $\xb_1,\dots,\xb_a$, each of length $K$, each subvector is encoded independently by multiplying with $\mathbf{C}$, and all the encoded vectors are stored in a single storage node, see Examples~1--3 of~\cite{PM_2018}.
In~\cite[Section~IV]{PM_2018}, a code is constructed for the case $t=1$ that achieves the codelength $l=2\ei$ using an $[m,m-2\ei]$ MDS code, which is guaranteed to exist if the field size $q \geq m$. 
In Remark~4 of \cite{PM_2018} the authors consider a modified system model for the function update problem which we show in Section~\ref{compare} of this paper to be equivalent to the case where $\Ab$ is striped with the number of stripes $a=t$. 
Construction~1 and Remark~4 of \cite{PM_2018} provide a code construction for this modified system model, and hence for the case $a=t$, that achieves codelength of $2t\ei$ over any field.

In this paper we provide a field-size aware characterization of the point-to-point function update problem. In particular, we provide bounds on the achievable codelength that take into account the effect of the field size and we provide constructions that trade-off the codelength for a smaller field size requirement. 
This is unlike the point-to-point results in \cite{PM_2018} which provide constructions only for the case $l=2\ei$ but assume that the field size $q$ is sufficiently large.
To the best of our knowledge, no prior analysis of this problem as a function of the field size $q$ is available in the literature except \cite{PM_2018} 
which assumes that the field size $q$ is large enough for a maximally recoverable code to exist.

We characterize the family of point-to-point function update problems where linear coding scheme is useful to save at least one transmission, i.e., $l \leq m-1$ is achievable (Theorem~\ref{thm3}, Section~\ref{cond}). 
This characterization is analyzed in terms of the covering radius of $\csr^\perp_A$, the dual of the code $\csr_A$, in Section~\ref{covering_rad}. We provide a lower bound (Theorem~\ref{thm4}, Section~\ref{low_bound}) and an upper bound (Theorem~\ref{up}, Section~\ref{code_const}) on optimal codelength based on linear error correcting codes. Similar to \cite{PM_2018} we also provide code constructions when $\Ab$ is striped (Section~\ref{t_1},\ref{t_ge_1}) but our focus is on the general case where $t \geq 1$ and $a \geq 1$. 
For the case when $t=1$ we provide a construction (Section~\ref{t_1}) which achieves the optimal codelength for the respective operating field size $q$, for any prime power $q \geq 2$. For the special case $q \geq m$ this code construction achieves codelength $2\ei$ and this matches the achieved codelength in Construction~2 of \cite{PM_2018} for $t=1$ which also requires $q \geq m$. Section~\ref{t_ge_1} provides code constructions for $t \geq 1$ using subspace codes and error correcting codes over field extensions. 
All these code constructions yield a trade-off between the chosen field size and achieved codelength where operating over a smaller field size results in a larger codelength than operating over a larger field size (for instance, see Example~\ref{exmp3}). 
When restricted to the special case $a=t$ our construction provides a valid coding scheme for the modified function update problem mentioned in \cite[Remark~4]{PM_2018} that matches codelength $2t\ei$ over any field $\Fb_q$ reported in \cite{PM_2018} (Section~\ref{compare}). 
The performance comparison of the constructed codes are discussed in Section~\ref{code_compare}. 
Finally, we show that the point-to-point function update problem is equivalent a \emph{functional index coding} or a \emph{generalized index coding} problem~\cite{DSS_2014,LDH_2015,GR_2016}.
Given a point-to-point function update problem we construct a functional index coding problem (Algorithm~1, Section~\ref{FIC_const}) such that a coding scheme is valid for the function update problem if and only if it is valid for the constructed functional index coding problem (Theorem~\ref{fu_cd_H}, Section~\ref{FIC_const}). 
This paper starts with describing the system model and providing relevant preliminary results in Section~\ref{sys}.

\emph{Notation}: Matrices and column vectors are denoted by bold uppercase and lowercase letters, respectively. For any positive integer $n$, the symbol $[n]$ denotes the set $\{1,\dots,n\}$. The Hamming weight of a vector $\xb$ is denoted as $\mathrm{wt}(\xb)$. The symbol $\Fb_q$ denotes the finite field of size $q$ and $\Fb_q^n$ denotes a column vector of $n$ elements over $\Fb_q$ where $q$ is a prime power. The $n \times n$ identity matrix is denoted as ${\bf I}_n$.  
 
\section{System Model and Preliminaries} \label{sys}
We consider a noiseless communication scenario with single transmitter and single receiver. The transmitter knows a column vector $\xb$ of $n$ information symbols where each information symbol is an element over finite field $\Fb_q$. The receiver stores the coded message $\Ab\xb \in \Fb_q^m$ where $\Ab \in \Fb_q^{m \times n}$ ($m \leq n$) and rank($\Ab$) = $m$. Now suppose the information symbol vector $\xb$ is updated to $\xb +\eb$ where $\eb$ is the update vector which is also a column vector of length $n$ over $\Fb_q$ with $\wtm(\eb)\leq \ei$, where $\wtm$ denotes the Hamming weight of a vector. The objective is to generate a codeword $\cb=(c_1,c_2,\dots,c_l)^T$ with codelength $l$ as small as possible such that the receiver can update its content to $\Ab(\xb+\eb)$ using the transmitted codeword $\cb$ and the older version of its content $\Ab\xb$. We assume the transmitter doesn't know about original information symbol vector $\xb$ or update vector $\eb$ but only knows the updated information symbol vector $(\xb+\eb)$. The problem of designing coding scheme to update the coded data $\Ab\xb$ available at the receiver to $\Ab(\xb+\eb)$ with $\wtm(\eb) \leq \ei$ will be called as $(\Ab,\ei)$ \textit{function update problem}.

\begin{mydef}
A valid encoding function of codelength $l$ for the $(\Ab,\ei)$ function update problem over the field $\Fb_q$ is a function
\begin{equation*}
\mathfrak{E}: \Fb_q^n~\longrightarrow~\Fb_q^l
\end{equation*}
such that there exists a decoding function $\mathfrak{D}: \Fb_q^l \times \Fb_q^m~\longrightarrow~\Fb_q^m$ satisfying the
following property: $\mathfrak{D}(\mathfrak{E}(\xb+\eb),\Ab\xb)=\Ab(\xb+\eb)$ for every $\xb \in \Fb_q^n$ and $\eb \in \Fb_q^n$ with $\wtm(\eb) \leq \ei$.
\end{mydef}

The objective of the code construction is to design a pair $(\mathfrak{E},\mathfrak{D})$ of encoding and decoding functions that minimizes the codelength $l$ and to calculate the optimal codelength over $\Fb_q$ which is the minimum codelength among all valid coding schemes.

A coding scheme $(\mathfrak{E},\mathfrak{D})$ is said to be linear if the encoding function is an $\Fb_q$-linear transformation. For a linear coding scheme, the codeword $\cb=\Hb(\xb+\eb)$, where $\Hb \in \Fb_q^{l \times n}$. The matrix $\Hb$ is the encoder matrix of the linear coding scheme. The minimum codelength among all valid linear coding schemes for the $(\Ab,\ei)$ function update problem over the field $\Fb_q$ will be denoted as $l_{q,\optm}$.

The trivial coding scheme that transmits the updated coded information symbols $\Ab(\xb+\eb)$  i.e., $\cb=\Ab(\xb+\eb)$ is a valid coding scheme with codelength $m$ since the receiver can directly update its content using $\cb$. We refer to this trivial coding scheme as \textit{naive scheme} where $\Hb=\Ab$. Thus, we have the following trivial upper bound on the optimum linear codelength
\begin{equation} \label{trivial_up}
l_{q,\optm} \leq m.
\end{equation}
In \cite{PM_2018} the authors provided a necessary and sufficient condition for a matrix $\Hb$ to be a valid encoder matrix for $(\Ab,\ei)$ function update problem. In Theorem 2 of \cite{PM_2018} the proof is given only for necessary condition for a matrix $\Hb$ to be a valid encoder matrix for $(\Ab,\ei)$ function update problem. For the sake of completeness here we first prove that the criterion~1 in \cite[Theorem~2]{PM_2018} is a necessary and sufficient condition for a matrix $\Hb$ to be a valid encoder matrix for $(\Ab,\ei)$ function update problem and then state the relevant results which will be helpful to derive other results of this paper. Let $\csr_{A}$ and $\csr_{H}$ denote the linear codes generated by the rows of $\Ab$ and $\Hb$ respectively.  Also let $\csr=\csr_A \cap \csr_H$ and let $\Pb$ be a generator matrix of $\csr$.
\begin{theorem} [Theorem 2, \cite{PM_2018}]\label{thm1}
A matrix $\Hb \in \Fb_q^{l \times n}$ is a valid encoder matrix for the $(\Ab,\ei)$ function update problem if and only if $\Pb\yb \neq \pmb{0}$ for any $\yb \in \Fb_q^n$ with $\wtm(\yb) \leq 2\ei$ and $\Ab\yb \neq \pmb{0}$.
\end{theorem} 
\begin{proof}
A matrix $\Hb \in \Fb_q^{l \times n}$ is a valid encoder matrix for the $(\Ab,\ei)$ function update problem if and only if  the receiver can uniquely determine $\Ab(\xb+\eb)$ from the received codeword $\Hb(\xb+\eb)$ and the side information
$\Ab\xb$. Hence for two pairs of information symbol vectors and update vectors $(\xb,\eb)$ and $(\xb',\eb')$ such that the coded information symbol vectors available at the receiver are identical i.e., $\Ab\xb=\Ab\xb'$ but updated coded information symbol vectors are distinct i.e., $\Ab(\xb+\eb) \neq \Ab(\xb'+\eb')$ then the transmitted codeword $\Hb(\xb+\eb)$ must be distinct from $\Hb(\xb'+\eb')$ to distinguish the two different updated coded information symbol vectors. Equivalently, the condition $\Hb(\xb+\eb) \neq \Hb(\xb'+\eb')$ should hold for every choice of $\xb,\xb',\eb,\eb' \in \Fb_q^n$ with $\wtm(\eb), \wtm(\eb') \leq \ei$ satisfying $\Ab\xb=\Ab\xb'$ and $\Ab(\xb+\eb) \neq \Ab(\xb'+\eb')$. Therefore $\Hb$ is a valid encoder matrix if and only if
\begin{equation*}
\Hb(\xb-\xb') \neq \Hb(\eb'-\eb)
\end{equation*}
for all $\xb,\xb' \in \Fb_q^n$ such that $\Ab\xb=\Ab\xb'$ and $\Ab(\xb-\xb') \neq \Ab(\eb'-\eb)$. Now denoting $\zb=\xb-\xb'$ and $\yb=\eb'-\eb$ we have
\begin{equation} \label{equ1}
\Hb\zb \neq \Hb\yb
\end{equation}
for all $\zb,\yb \in \Fb_q^n$ and $\wtm(\yb)=\wtm(\eb'-\eb) \leq 2\ei$ such that $\Ab\zb=\pmb{0}$ and $\Ab\zb \neq \Ab\yb$. Now reformulating the condition given in (\ref{equ1}) we obtain $\Hb(\zb-\yb) \neq \pmb{0}$ for all $\zb,\yb \in \Fb_q^n$ that satisfy $\wtm(\yb)\leq 2\ei$, $\Ab\zb=\pmb{0}$ and $\Ab\yb \neq \pmb{0}$. Therefore $\Hb$ is a valid encoder matrix if and only if for all $\zb,\yb \in \Fb_q^n$ and $\wtm(\yb)\leq 2\ei$ if $\zb \in \csr_A^{\perp}$ and $\yb \notin \csr_A^{\perp}$ then $(\yb-\zb) \notin \csr_H^{\perp}$. Hence $\Hb$ is a valid encoder matrix if and only if for all $\yb \in \Fb_q^n$ with $\wtm(\yb)\leq 2\ei$ if  $\yb \notin \csr_A^{\perp}$ then $\yb \notin \csr_A^{\perp}+\csr_H^{\perp}$. Now using the fact that $\csr_A^{\perp}+\csr_H^{\perp}=(\csr_A \cap \csr_H)^{\perp}=\csr^{\perp}$ we deduce that $\Hb$ is a valid encoder matrix if and only if for all $\yb \in \Fb_q^n$ with $\wtm(\yb)\leq 2\ei$ such that $\yb \notin \csr_A^{\perp}$ also satisfies $\yb \notin \csr^{\perp}$. Hence the statement of the theorem follows.
\end{proof}
\begin{lemma} [Remark 2, \cite{PM_2018}] \label{lmm1}
Let $\Hb \in \Fb_q^{l \times n}$ be a valid encoder matrix for the $(\Ab,\ei)$ function update problem. Let $\Pb$ be a generator matrix of the code $\csr=\csr_A \cap \csr_H$. Then $\Pb$ is also a valid encoder matrix for the $(\Ab,\ei)$ function update problem. 
\end{lemma}

If we consider a valid encoder matrix $\Hb' \in \Fb_q^{l' \times n}$ such that $\csr_{H'} \nsubseteq \csr_A$, then we can find another valid encoder matrix $\Hb \in \Fb_q^{l \times n}$ as the generator matrix of the code $\csr_A \cap \csr_{H'}$. Since $\csr_H$ is a subcode of $\csr_A \cap \csr_{H'}$, we have $l > l'$. Therefore  the encoder matrix $\Hb'$ has sub-optimal codelength.  So from now we only consider encoder matrices $\Hb$ such that $\csr_H \subseteq \csr_A$. Since we assume $\csr_H \subseteq \csr_A$ we can write $\Hb=\Sb\Ab$ for some matrix $\Sb \in \Fb_q^{l \times m}$.

Now using $\Pb=\Hb$ we restate Theorem~\ref{thm1} as follows. A matrix $\Hb \in \Fb_q^{l \times n}$ such that $\csr_H \subseteq \csr_A$ is a valid encoder matrix for the $(\Ab,\ei)$ function update problem if and only if for any $\yb \in \Fb_q^n$ with $\wtm(\yb) \leq 2\ei$ and $\Ab\yb \neq \pmb{0}$ satisfies $\Hb\yb \neq \pmb{0}$. We define the collection $\ical(\Ab,\ei)$ as the set of all vectors $\yb \in \Fb_q^n$ with $\wtm(\yb) \leq 2\ei$ such that $\Ab\yb \neq \pmb{0}$ i.e.,
\begin{equation} \label{ical_fu}
\ical(\Ab,\ei)=\{\yb \in \Fb_q^n~|~\Ab\yb \neq \pmb{0},~0 < \wtm(\yb) \leq 2\ei\}.
\end{equation}
\begin{theorem} \label{thm2}
A matrix $\Hb=\Sb\Ab$ for some matrix $\Sb \in \Fb_q^{l \times m}$ is a valid encoder matrix for the $(\Ab,\ei)$ function update problem if and only if 
\begin{equation*}
\Hb\yb \neq \pmb{0},~~~\forall \yb \in \ical(\Ab,\ei).
\end{equation*}
\end{theorem}
 Now we define the collection $\ical_{\text{FU}}(\Ab,\ei)$ as the set of all non-zero linear combinations of $2\ei$ or fewer columns of $\Ab$ over $\Fb_q$ i.e.,
\begin{equation*}
\ical_{\text{FU}}(\Ab,\ei) = \{\Ab\yb~|~0 < \wtm(\yb) \leq 2\ei\}\backslash \{\pmb{0}\}=\{\Ab\yb~|~\yb \in \ical(\Ab,\ei)\}.
\end{equation*} 
Note that $|\ical_{\text{FU}}| \leq q^m-1$ since $\pmb{0} \notin \ical_{\text{FU}}$.
\begin{corollary} \label{corr1}
$\Hb=\Sb\Ab$ is a valid encoder matrix for the $(\Ab,\ei)$ function update problem if and only if
\begin{equation*}
\Sb\zb \neq \pmb{0},~~~\forall \zb \in \ical_{\text{FU}}(\Ab,\ei).
\end{equation*} 
\end{corollary} 
\section{Necessary and Sufficient Condition for $l_{q,\optm} < m$}  \label{cond}
In this section we will characterize the family of point-to-point function update problems where linear coding is useful to save at least one transmission compared to the naive scheme i.e., $l_{q,\optm} <m$. First we will derive some preliminary results which will be helpful to derive the main result of this section.
\begin{lemma} \label{lmm2}
The collection $\ical_{\text{FU}}(\Ab,\ei)$ is closed under non-zero scalar multiplication.
\end{lemma}
\begin{proof}
Suppose $\zb \in \ical_{\text{FU}}$. There exists a $\yb \in \Fb_q^n$ with $0< \wtm(\yb) \leq 2\ei$ such that $\zb=\Ab\yb$. For any $\alpha \in \Fb_q^{\ast}$, $\alpha\zb=\alpha\Ab\yb=\Ab(\alpha\yb)=\Ab\yb'$, where $\yb'=\alpha\yb$. Now as $0 < \wtm(\yb) \leq 2\ei$, it follows that $0 < \wtm(\yb') \leq 2\ei$. Again $\Ab\yb'=\Ab(\alpha\yb)=\alpha \Ab\yb \neq \pmb{0}$ as $\Ab\yb \in \ical_{\text{FU}}$ and $\alpha \neq 0$. Therefore for any $\alpha \in \Fb_q^{\ast}$, $\alpha\zb \in \ical_{\text{FU}}$. Hence the lemma holds.
\end{proof} 
\subsection{A coding scheme for a family of $(\Ab,\ei)$ function update problems} \label{coding_scheme}
Consider any $(\Ab,\ei)$ function update problem where there exists a non-zero $\ub \in \Fb_q^m$ such that $\ub \notin \ical_{\text{FU}}$. Let $\csr_u$ be the subspace of $\Fb_q^m$ generated by $\ub$. Therefore dim$(\csr_u)=1$. Note that dim$(\csr_u^{\perp})=m-1$. Let $\Sb \in \Fb_q^{(m-1) \times m}$ be a generator matrix of the code $\csr_u^{\perp}$. The matrix $\Sb$ is a parity check matrix of the code $\csr_u$.

\begin{lemma} \label{lmm3}
The matrix $\Sb$ satisfies $\Sb\zb \neq \pmb{0}$ for all $\zb \in \ical_{\text{FU}}(\Ab,\ei)$.
\end{lemma}
\begin{proof}
Proof by contradiction. Let there exist a $\zb \in \ical_{\text{FU}}$ such that $\Sb\zb =\pmb{0}$. This implies that $\zb \in \csr_u$. Therefore there exists an $\alpha \in \Fb_q^{\ast}$ such that $\zb=\alpha\ub$. Now as $\alpha \in \Fb_q^{\ast}$, $\alpha^{-1}$ exists and hence $\ub=\alpha^{-1}\zb$. Now as $\zb \in \ical_{\text{FU}}$ and $\ical_{\text{FU}}$ is closed under non-zero scalar multiplication (using Lemma~\ref{lmm3}), $\ub \in \ical_{\text{FU}}$ which is a contradiction. Hence the lemma holds. 
\end{proof}
Now using Corollary~\ref{corr1}, we obtain a valid encoder matrix for the $(\Ab,\ei)$ function update problem over $\Fb_q$ as $\Hb=\Sb\Ab$ with codelength $l=m-1$, whenever there exists a non-zero vector in $\Fb_q^m \backslash \ical_{\text{FU}}$. We do not claim that this coding scheme yields the optimal codelength $l_{q,\optm}$.
\begin{example} \label{exmp1}
Consider the $(\Ab,1)$ function update problem over binary field $\Fb_2$ where $m=5$, $n=8$, $\epsilon=1$ and the matrix $\Ab \in \Fb_2^{5 \times 8}$ is given by
\begin{equation*}
\Ab=
\begin{bmatrix}
1 & 1 & 0 & 1 & 0 & 0 & 1 & 1\\
0 & 0 & 1 & 0 & 0 & 1 & 0 & 0\\
1 & 0 & 0 & 1 & 0 & 0 & 1 & 1\\
0 & 0 & 1 & 0 & 1 & 1 & 1 & 1\\
1 & 0 & 1 & 1 & 0 & 0 & 1 & 1
\end{bmatrix}.
\end{equation*} 
Note that rank$(\Ab)=5$ over $\Fb_2$. The non-zero vector $\ub=[0~1~1~0~0] \in \Fb_2^5$ satisfies $\ub \notin \ical_{\text{FU}}(\Ab,1)$. The parity check matrix of the code $\csr_u$, generated by $\ub$ is given by
\begin{equation*}
\Sb=
\begin{bmatrix}
1 & 0 & 0 & 1 & 1\\
0 & 1 & 1 & 0 & 0\\
0 & 0 & 0 & 1 & 1\\
1 & 1 & 1 & 1 & 0\\
\end{bmatrix}.
\end{equation*} 
Therefore we obtain a valid encoder matrix $\Hb \in \Fb_2^{4 \times 8}$ with codelength $l=4$ for the function update problem as 
\begin{equation*}
\Hb= \Sb\Ab=
\begin{bmatrix}
0 & 1 & 0 & 0 & 1 & 1 & 1 & 1\\
1 & 0 & 1 & 1 & 0 & 1 & 1 & 1\\
1 & 0 & 0 & 1 & 1 & 1 & 0 & 0\\
0 & 1 & 0 & 0 & 1 & 0 & 1 & 1
\end{bmatrix}.
\end{equation*}   
\end{example}
Now we derive a necessary and sufficient condition for any $(\Ab,\ei)$ function update problem to save at least one transmission using linear coding scheme compared to the naive scheme.
\begin{theorem}  \label{thm3}
For an $(\Ab,\ei)$ function update problem, $l_{q,\optm}=m$ if and only if $|\ical_{\text{FU}}(\Ab,\ei)|=q^m-1$.
\end{theorem}
\begin{proof}
To prove the theorem we first show that for any $(\Ab,\ei)$ function update problem if $|\ical_{\text{FU}}(\Ab,\ei)| = (q^m-1)$ then \mbox{$l_{q,\optm}=m$}. Next we show that if \mbox{$|\ical_{\text{FU}}(\Ab,\ei)| < (q^m-1)$} then \mbox{$l_{q,\optm} \leq (m-1)$}.
\vspace{2mm}\\
\textit{Proof of first part i.e., if $|\ical_{\text{FU}}(\Ab,\ei)| = (q^m-1)$ then  $l_{q,\optm} = m$}: 

Let $\Hb$ be an optimal encoder matrix with $l=l_{q,\optm}$. Then there exists a matrix $\Sb \in \Fb_q^{l_{q,\optm} \times m}$ such that $\Hb=\Sb\Ab$. From Corollary~\ref{corr1} we obtain $\Sb\zb \neq \pmb{0}$ for all $\zb \in \ical_{\text{FU}}$.  Since $\ical_{\text{FU}}$ contains all non-zero vectors from $\Fb_q^m$, the columns of $\Sb$ are linearly independent. Hence $l_{q,\optm} \geq m$. Again from (\ref{trivial_up}), we have $l_{q,\optm} \leq m$. Hence $l_{q,\optm}=m$.  
\vspace{2mm}\\
\textit{Proof of second part i.e., if $|\ical_{\text{FU}}(\Ab,\ei)| < (q^m-1)$ then  $l_{q,\optm} \leq (m-1)$}: 

If $|\ical_{\text{FU}}(\Ab,\ei)| < (q^m-1)$ then there exists a non-zero vector $\ub \in \Fb_q^m$ such that $\ub \notin \ical_{\text{FU}}(\Ab,\ei)$. Therefore using the technique described in Section~\ref{coding_scheme} we can construct a valid encoder matrix such that we save one transmission compared to the naive scheme, i.e., $l_{q,\optm} \leq (m-1)$. Hence the lemma holds.
\end{proof}
Now we provide a sufficient condition on the field size $q$ to save at least one transmission compared to the naive scheme for any $(\Ab,\ei)$ function update problem. 


\begin{corollary} \label{corr2}
For any $\Ab \in \Fb_q^{m \times n}$ with rank$(\Ab)=m$ where $m > 2\ei$, 
\begin{center}
$l_{q,\optm} \leq (m-1)$ if $q \geq \binom{n}{2\ei}^{1/(m-2\ei)}$
\end{center}.
\end{corollary}
\begin{proof}
If $q \geq \binom{n}{2\ei}^{1/(m-2\ei)}$ then
\begin{align*}
q^{m-2\ei} &\geq \binom{n}{2\ei}\\
\text{or,}~~\frac{q^m}{q^{2\ei}} &\geq \binom{n}{2\ei}.
\end{align*}
Using the fact that if $a > b$ then $\frac{a-1}{b-1} > \frac{a}{b}$, we have
\begin{align*}
\frac{q^m-1}{q^{2\ei}-1} &> \frac{q^m}{q^{2\ei}}~~~(\text{since $q^m>q^{2\ei}$})\\
\text{or,}~~\frac{q^m-1}{q^{2\ei}-1} &> \binom{n}{2\ei}\\
\text{or,}~~q^m-1 &> \binom{n}{2\ei}(q^{2\ei}-1).
\end{align*}
Now for any $\Ab \in \Fb_q^{m \times n}$ with rank$(\Ab)=m$, the number of distinct non-zero linear combinations of $2\ei$ or fewer columns of $\Ab$ is at the most $\binom{n}{2\ei}(q^{2\ei}-1)$. Therefore $|\ical_{\text{FU}}| \leq \binom{n}{2\ei}(q^{2\ei}-1)$. Hence $(q^m-1) > \binom{n}{2\ei}(q^{2\ei}-1) \geq |\ical_{\text{FU}}|$. Now using Theorem~\ref{thm3} we have $l_{q,\optm} \leq (m-1)$.
\end{proof}
\vspace{1mm}
\subsection{Relation with covering radius}  \label{covering_rad}
The covering radius of an $[n,k]$ linear code $\csr$ over $\Fb_q$, denoted by $\rcovm(\csr)$, is defined as the smallest integer $r$ such that the spheres of radius $r$ centered at each codeword of $\csr$ cover the whole space $\Fb_q^n$. We can determine covering radius of a linear code in terms of the cosets of the code. For any vector $\ab \in \Fb_q^n$, the set $\ab+\csr=\{\ab+\cb~|~\cb \in \csr\}$ is called a coset of the code $\csr$ and in any coset, a vector with minimum Hamming weight is called a coset leader. The covering radius $\rcovm(\csr)$ of the code $\csr$ is the largest among the Hamming weight of all the coset leaders. Upon denoting $\Hb'$ as a parity check matrix of $\csr$, $\Hb'\ub$ is the syndrome of the vector $\ub \in \Fb_q^n$. Two vectors $\ub,\textbf{v} \in \Fb_q^n$ have the same syndrome if and only if they belong to the same coset of $\csr$. Hence there is an one-to-one correspondence between syndromes and cosets \cite{cov_code}.  

For the $(\Ab,\ei)$ function update problem let $\csr_A$ be the linear code generated by $\Ab$. Hence $\Ab$ is a parity check matrix of the code $\csr_A^{\perp}$ which is the dual code of $\csr_A$. Now considering a vector $\zb \in \ical_{\text{FU}}$, $\zb$ can be expressed as $\Ab\yb$ where $\yb \in \Fb_q^n$ with $0 < \wtm(\yb) \leq 2\ei$. Therefore the vector $\zb$ denotes the syndrome of a vector $\yb \in \Fb_q^n$ with $0 < \wtm(\yb) \leq 2\ei$ that belongs to some coset of $\csr_A^{\perp}$. Note that any vector that belongs to $\ical_{\text{FU}}$ is non-zero, hence can not be the syndrome of the codewords of $\csr^\perp_A$. Note that $\zb$ is the syndrome of the coset leader of the coset $\yb +\csr_A^{\perp}$. Since $\yb$ is a vector that belongs to the coset and $0 < \wtm(\yb) \leq 2\ei$, the Hamming weight of the coset leader of the coset is at the most $2\ei$.   

\begin{corollary} \label{corr3}
For an $(\Ab,\ei)$ function update problem, $l_{q,\optm} = m$ if and only if $\rcovm(\csr_A^{\perp}) \leq 2\ei$.
\end{corollary}
\begin{proof}
From Theorem~\ref{thm3} we have $l_{q,\optm}=m$ if and only if $|\ical_{\text{FU}}|=q^m-1$. Hence to prove the corollary we prove that for any $(\Ab,\ei)$ function update problem $|\ical_{\text{FU}}|=q^m-1$ if and only if $\rcovm(\csr_A^{\perp}) \leq 2\ei$. 
\vspace{2mm}\\
\textit{Proof of $\rcovm(\csr_A^{\perp}) \leq 2\ei$ if $|\ical_{\text{FU}}|=q^m-1$:} Since the collection $\ical_{\text{FU}}$ contains all non-zero vectors over $\Fb_q^m$, each non-zero vector $\zb \in \Fb_q^m$ is the syndrome of some vector $\yb \in \Fb_q^n$ with $0 < \wtm(\yb) \leq 2\ei$ that belongs to some coset of $\csr_A^{\perp}$. Since there exists a one-to-one correspondence between the syndromes and cosets, for each vector $\zb \in \ical_{\text{FU}}$ there exists a coset of $\csr_A^{\perp}$ that contains a vector $\yb$ with $0 < \wtm(\yb) \leq 2\ei$. Hence the coset leader of each coset has Hamming weight at the most $2\ei$. Therefore the largest Hamming weight of the coset leaders among all cosets of $\csr_A^{\perp}$ is at the most $2\ei$. Hence $\rcovm(\csr_A^{\perp}) \leq 2\ei$.
\vspace{2mm}\\
\textit{Proof of $|\ical_{\text{FU}}|=q^m-1$ if $\rcovm(\csr_A^{\perp}) \leq 2\ei$:} Since $\rcovm(\csr_A^{\perp}) \leq 2\ei$, the largest Hamming weight of the coset leaders among all cosets of $\csr_A^{\perp}$ is at the most $2\ei$. Since there exists a one-to-one correspondence between the syndromes and cosets, each syndrome $\zb \in \Fb_q^m$ can be expressed as $\Ab\yb$ for some coset leader $\yb \in \Fb_q^n$ satisfies $0 <\wtm(\yb) \leq 2\ei$. We know that the syndromes of a particular linear code covers the whole space. Hence any vector $\zb \in \Fb_q^m \backslash \{\pmb{0}\}$ can be expressed as $\Ab\yb$ for some $\yb \in \Fb_q^n$ with $0 <\wtm(\yb) \leq 2\ei$. Since $\ical_{\text{FU}}$ consists only non-zero vectors that satisfies the above property, $|\ical_{\text{FU}}|=q^m-1$.

Hence $l_{q,\optm} = m$ if and only if $\rcovm(\csr_A^{\perp}) \leq 2\ei$.
\end{proof}
\begin{example} \label{exmp2}
In this example we calculate the minimum number of rows of $\Ab_{m \times n}$ such that $l_{q,\optm} \leq (m-1)$ is guaranteed for $q=2$, $\ei=1$ and $n=8$. Now $l_{q,\optm} \leq (m-1)$ if and only if $\rcovm(\csr_A^{\perp}) \geq 2\ei+1=3$. From Table I of \cite{GS_1985} we observe that for any binary code of length $8$ and dimension up to $3$, covering radius is at least $3$. Thus dim$(\csr_A^{\perp}) \geq 3$ implies $l_{q,\optm} \leq (m-1)$. Hence $(n-m) \leq 3$ and $m \geq (n-3)=5$. Therefore for any matrix $\Ab \in \Fb_2^{5 \times 8}$ with rank$(\Ab)=5$ we can save one transmission compared to the naive scheme. One such example of $\Ab$ is given in Example~\ref{exmp1}.
\end{example}
\section{Lower Bound on Optimal Codelength}  \label{low_bound}
In this section we derive a lower bound on the optimal codelength $l_{q,\optm}$ over $\Fb_q$. First we derive two preliminary lemmas which will help to derive the lower bound. 
\begin{lemma} \label{equiv}
For any $(\Ab,\ei)$ function update problem and for any invertible matrix $\Kb \in \Fb_q^{m \times m}$, $\ical(\Ab,\ei)=\ical(\Kb\Ab,\ei)$.  
\end{lemma}
\begin{proof}
To prove the lemma we first show that $\ical(\Ab,\ei)\subseteq \ical(\Kb\Ab,\ei)$ and then $\ical(\Kb\Ab,\ei) \subseteq \ical(\Ab,\ei)$.
\vspace{2mm}\\
\textit{Proof for $\ical(\Ab,\ei)\subseteq \ical(\Kb\Ab,\ei)$}: Suppose $\yb \in \ical(\Ab,\ei)$. Then from (\ref{ical_fu}) we have $\Ab\yb \neq \pmb{0}$. Now left multiplying both side by $\Kb$ we obtain $\Kb\Ab\yb \neq \pmb{0}$ since $\Kb$ is invertible. Hence $\yb \in \ical(\Kb\Ab,\ei)$.
\vspace{2mm}\\
\textit{Proof for $\ical(\Kb\Ab,\ei)\subseteq \ical(\Ab,\ei)$}: Suppose $\yb \in \ical(\Kb\Ab,\ei)$. Then from (\ref{ical_fu}) we have $\Kb\Ab\yb \neq \pmb{0}$. Since $\Kb$ is invertible, $\Kb^{-1}$ exists. Now left multiplying both side by $\Kb^{-1}$ we obtain $\Ab\yb \neq \pmb{0}$. Hence $\yb \in \ical(\Ab,\ei)$.

Hence the lemma holds.
\end{proof}
For any $(\Ab,\ei)$ function update problem $\Ab \in \Fb_q^{m \times n}$ with rank($\Ab$)=$m$. Hence $\Ab$ contains $m$ linearly independent columns. Now consider a matrix $\Kb'$ which contains $m$ linearly independent columns of $\Ab$. Note that $\Kb'$ is an $m \times m$ full rank matrix and hence invertible. Denote $\Kb=\Kb'^{-1}$ and $\Ab'=\Kb\Ab$. From Lemma~\ref{equiv} we observe that $(\Ab,\ei)$ and $(\Ab',\ei)$ are equivalent function update problems and any matrix $\Hb$ is a valid encoder matrix of $(\Ab,\ei)$ function update problem if and only if $\Hb$ is a valid encoder matrix of $(\Ab',\ei)$ function update problem. Hence we conclude that the linear code generated by the rows of $\Hb$ is a subcode of the linear code generated by the rows of $\Ab'$ i.e., $\csr_{H} \subseteq \csr_{A'}$. Hence there exists a matrix $\Sb' \in \Fb_q^{l \times m}$ such that $\Hb=\Sb'\Ab'$. Now using the equivalence between $(\Ab,\ei)$ and $(\Ab',\ei)$ function update problems and using Corollary~\ref{corr1} we say that $\Hb=\Sb'\Ab'$ is a valid encoder matrix of the $(\Ab,\ei)$ function update problem if and only if $\Sb'\zb \neq \pmb{0}$ for all $\zb \in \ical_{\text{FU}}(\Ab',\ei)$.

Let $\bcal^{\ast}(m,2\ei)=\{\zb \in \Fb_q^m~|~0 < \wtm(\zb) \leq 2\ei\}$ be the set of all non-zero vectors in $\Fb_q^m$ of Hamming weight at the most $2\ei$.
\begin{lemma} \label{lmm6}
For any $(\Ab,\ei)$ function update problem
\begin{center}
$\bcal^{\ast}(m,2\ei) \subseteq \ical_{\text{FU}}(\Ab,\ei)$.
\end{center}
\end{lemma}
\begin{proof}
For an $(\Ab,\ei)$ function update problem, $\Ab'=\Kb\Ab$ where $\Kb=\Kb'^{-1}$ and $\Kb'$ consists of $m$ linearly independent columns of $\Ab$. Note that the sub-matrix of $\Ab'$ that contains the corresponding columns forms an $m \times m$ identity matrix. Now if we consider any non-zero linear combination of $2\ei$ of fewer columns of this sub-matrix we obtain all non-zero vectors over $\Fb_q^m$ with Hamming weight at the most $2\ei$. Hence $\bcal^{\ast}(m,2\ei) \subseteq \ical_{\text{FU}}(\Ab',\ei)=\ical_{\text{FU}}(\Ab,\ei)$. The last equality holds due to Lemma~\ref{equiv}.
\end{proof}
Let $k_q(m,2\ei+1)$ be the maximum dimension among all linear codes over $\Fb_q$ with blocklength $m$ and minimum distance $d_{\mathrm{min}} \geq 2\ei+1$.
\begin{theorem} \label{thm4}
The optimal codelength of the $(\Ab,\ei)$ function update problem over $\Fb_q$ satisfies
\begin{center}
$l_{q,\optm} \geq m-k_q(m,2\ei+1)$.
\end{center}
\end{theorem}
\begin{proof}
Let $\Hb$ be an optimal encoder matrix of $(\Ab,\ei)$ function update problem with codelength $l=l_{q,\optm}$. Then there exists a matrix $\Sb \in \Fb_q^{l_{q,\optm} \times m}$ such that $\Hb=\Sb\Ab$. Now using Corollary~\ref{corr1} we have $\Sb\zb \neq \pmb{0}$ for all $\zb \in \ical_{\text{FU}}(\Ab,\ei)$. Since $\bcal^{\ast}(m,2\ei) \subseteq \ical_{\text{FU}}(\Ab,\ei)$, it follows that $\Sb\zb \neq \pmb{0}$ for all $\zb \in \bcal^{\ast}(m,2\ei)$. Therefore any set of $2\ei$ columns of $\Sb$ are linearly independent. Hence $\Sb$ is a parity check matrix of a linear code of block length $m$ and minimum distance at least $2\ei+1$. Thus the dimension of this code satisfies $m-l_{q,\optm} \leq k_q(m,2\ei+1)$. Then $l_{q,\optm} \geq m-k_q(m,2\ei+1)$.
\end{proof}

Theorem~\ref{thm4} provides a lower bound that is aware of the field size $q$. 
This is tighter than the bound $l \geq 2\ei$ given in~\cite{NSR_2014,NSRSR_update_2014,PM_2018} since from Singleton bound we know that $k_q(m,2\ei+1) \leq m - 2\ei$, and this combined with Theorem~\ref{thm4} yields $l \geq 2\ei$. 
Hence, irrespective of the matrix $\Ab$, a necessary condition for $l_{q,\optm}=2\ei$ is that an $[m,m-2\ei]$ MDS code over $\Fb_q$ must exist.

\section{Code constructions}    \label{code_const}
In this section we first derive an upper bound on the optimal codelength $l_{q,\optm}$ over $\Fb_q$ and then provide code constructions for $(\Ab,\ei)$ function update problem when $\Ab$ is in form given by (\ref{a_struct}). Define $\eta=\max\limits_{\zb \in \ical_{\text{FU}}}{\wtm(\zb)}$.
\begin{theorem} \label{up}
The optimal codelength of the $(\Ab,\ei)$ function update problem over $\Fb_q$ satisfies
\begin{center}
$l_{q,\optm} \leq m-k_q(m,\eta+1)$.
\end{center}
\end{theorem}
\begin{proof}
From Corollary~\ref{corr1} we have that a matrix $\Hb=\Sb\Ab \in \Fb_q^{l \times n}$ for some matrix $\Sb \in \Fb_q^{l \times m}$, is a valid encoder matrix if and only if $\Sb\zb \neq \pmb{0},~\forall \zb \in \ical_{\text{FU}}$. To satisfy this condition it is sufficient that any set of $\eta$ columns of $\Sb$ are linearly independent. Now consider $\Sb$ as a parity check matrix of the largest linear code with blocklength $m$ and minimum distance $d_{\mathrm{min}} \geq \eta+1$. The resulting codelength $l=m-k_q(m,\eta+1)$. Hence the upper bound on the optimal codelength holds.
\end{proof}
\subsection{Code constructions for striped data} \label{striped_data}
In this section we provide linear code construction of an $(\Ab^S,\ei)$ function update problem where $\Ab^S \in \Fb_q^{m \times n}$ follows the structure given by
\begin{equation} \label{a_struct}
\Ab^S=
\begin{bmatrix}
\Cb & \pmb{0} & \dots & \pmb{0}\\
\pmb{0} & \Cb & \dots & \pmb{0}\\
\vdots & \vdots & \ddots & \vdots\\
\pmb{0} & \pmb{0} & \dots & \Cb
\end{bmatrix}
\end{equation}
where $\Cb \in \Fb_q^{t \times K}$ and $\pmb{0}$ is a $t \times K$ matrix over $\Fb_q$ whose all elements are $0$. Let $a$ be the number of repetitions of $\Cb$ in the matrix $\Ab^S$. Hence we write $m=at$ and $n=aK$. First we consider the family of $(\Ab^S,\ei)$ function update problems where $t=1$ and show that for this case the lower bound on optimal codelength given in Theorem~\ref{thm4} and the upper bound on optimal codelength given in Theorem~\ref{up} exactly matches with each other. Hence we characterize the optimal codelength for this family of function update problems. Our code construction is based on an appropriately chosen linear error correcting code. Note that in Section IV of \cite{PM_2018} the authors provided a linear code construction based on maximally recoverable subcodes (MRSC) which requires field size $q \geq m$ and uses an $[m,m-2\ei]$ MDS code. In comparison our code construction is suitable for any field size. 
\vspace{2mm}
\subsubsection{Code Constructions for the family of $(\Ab^S,\ei)$ function update problems with $t=1$} \label{t_1} In this sub-section we first calculate the optimal codelength for such family of function update problems and then provide a code construction based on an appropriately chosen linear error correcting code.
\begin{theorem}  \label{opt_len}
For the family of $(\Ab^S,\ei)$ function update problems with $t=1$ the optimal codelength over $\Fb_q$ is given by 
\begin{equation*}
l_{q,\optm}=m-k_q(m,2\ei+1).
\end{equation*}
\end{theorem}
\begin{proof}
Consider any $\zb \in \ical_{\text{FU}}(\Ab^S,\ei)$. Then $\zb$ can be written as $\zb=\Ab^S\yb$ for some $\yb \in \Fb_q^n$ with $0<\wtm(\yb)\leq 2\ei$. Hence we write $\zb=[\Ab^{S,1}~\Ab^{S,2}~\dots~\Ab^{S,n}]\yb$ where $\Ab^{S,i}$ denotes the $i^{\thm}$ column of $\Ab^S$ and $\wtm(\Ab^{S,i})=1$ for all $i \in [n]$. Now $\wtm(\zb)=\wtm(\Ab^{S,1}y_1+\Ab^{S,2}y_2+\dots+\Ab^{S,n}y_n) \leq \wtm(\Ab^{S,1}y_1)+\wtm(\Ab^{S,2}y_2)+\dots+\wtm(\Ab^{S,n}y_n)$. Since $0 < \wtm(\yb) \leq 2\ei$, at the most $2\ei$ terms among $\Ab^{S,1}y_1,\Ab^{S,2}y_2,\dots,\Ab^{S,n}y_n$ are non-zero and each $\Ab^{S,i}y_i,~i \in [n]$ has Hamming weight at the most $1$. Hence for any $\zb \in \ical_{\text{FU}}(\Ab^S,\ei)$, we have $\wtm(\zb) \leq 2\ei$. It is easy to observe that $\eta=\max\limits_{\zb \in \ical_{\text{FU}}(\Ab^S,\ei)}{\wtm(\zb)}=2\ei$. So using Theorem~\ref{up} we have $l_{q,\optm} \leq m-k_q(m,2\ei+1)$. Again from Theorem~\ref{thm4} we have $l_{q,\optm} \geq m-k_q(m,2\ei+1)$. Since the lower bound and the upper bound matches with each other we have $l_{q,\optm} = m-k_q(m,2\ei+1)$.  
\end{proof}  
Now we provide a code construction for the family of $(\Ab^S,\ei)$ function update problems with $t=1$. Since for any $(\Ab^S,\ei)$ function update problem with $t=1$ the value of $\eta$ is  $2\ei$, it is sufficient that $\Sb\zb \neq \pmb{0}$ for any $\zb$ with $0 < \wtm(\zb)\leq 2\ei$. Hence it is sufficient that any $2\ei$ columns of $\Sb$ are linearly independent. Now consider $\Sb$ as a parity check matrix of a linear code of maximum dimension with blocklength $m$ and minimum distance $d_{\text{min}} \geq 2\ei+1$ and set the encoder matrix $\Hb=\Sb\Ab^S$. This code achieves the optimal codelength $l_{q,\optm}=m-k_q(m,2\ei+1)$. Now if $q \geq m$ then there exists an MDS code over $\Fb_q$ with blocklength $m$ and minimum distance $d_{\text{min}}=2\ei+1$ which has maximum dimension $k_q(m,2\ei+1)=m-2\ei$ among all linear codes over $\Fb_q$. Hence choosing $\Sb$ as a parity check matrix of an $[m,m-2\ei]$ MDS code $\Fb_q,~q \geq m$ and encoder matrix $\Hb=\Sb\Ab^S$ we achieve codelength $l_{q,\optm}=2\ei$ which matches the codelength achieved by the construction given in Section IV of \cite{PM_2018} which also requires $q \geq m$.
\begin{example} \label{exmp3}
Consider an $(\Ab^S,\ei)$ function update problem over $\Fb_2$ where $\ei=1$ and $\Ab^S$ is given by
\begin{equation*} 
\Ab^S=
\setcounter{MaxMatrixCols}{20}
\begin{bmatrix}
1 & 1 & 1 & 0 & 0 & 0 & 0 & 0 & 0 & 0 & 0 & 0\\  
0 & 0 & 0 & 1 & 1 & 1 & 0 & 0 & 0 & 0 & 0 & 0\\
0 & 0 & 0 & 0 & 0 & 0 & 1 & 1 & 1 & 0 & 0 & 0\\
0 & 0 & 0 & 0 & 0 & 0 & 0 & 0 & 0 & 1 & 1 & 1
\end{bmatrix}.
\end{equation*}
Now from \cite{Grassl:codetables} we have $k_2(4,3)=1$. Hence choosing $\Sb$ as parity check matrix of a $[4,1]$ repetition code over $\Fb_2$ we achieve codelength $l_{2,\optm}=3$.

If we view the above matrix $\Ab^S$ as over $\Fb_4$ and the $(\Ab^S,\ei)$ function update problem over $\Fb_4$ where $\ei=1$, from \cite{Grassl:codetables} we have $k_4(4,3)=2$. Hence choosing $\Sb$ as parity check matrix of a $[4,2,3]$ MDS code over $\Fb_4$ we achieve codelength $l_{4,\optm}=2$.
\end{example}
\vspace{2mm}
\subsubsection{Code constructions for the family of $(\Ab^S,\ei)$ function update problems where $t\geq 1$} \label{t_ge_1} 
In this sub-section we provide a linear code construction for the family of $(\Ab^S,\ei)$ function update problems where $\Ab^S$ is given in (\ref{a_struct}) with $t\geq 1$. A matrix $\Hb \in \Fb_q^{l \times n}$ is a valid encoder matrix if and only if there exists a matrix $\Sb \in \Fb_q^{l \times m}$ such that $\Hb=\Sb\Ab^S$ satisfies $\Sb\zb \neq \pmb{0}$ for all $\zb \in \ical_{\text{FU}}(\Ab^S,\ei)$. For any vector $\zb \in \ical_{\text{FU}}(\Ab^S,\ei)$ we write $\zb=\Ab\yb$ for some $\yb \in \Fb_q^n$ with $0 < \wtm(\yb) \leq 2\ei$. Hence
\begin{equation*}
\zb = \Ab\yb =
\begin{bmatrix}
\Cb & \pmb{0} & \dots & \pmb{0}\\
\pmb{0} & \Cb & \dots & \pmb{0}\\
\vdots & \vdots & \ddots & \vdots\\
\pmb{0} & \pmb{0} & \dots & \Cb
\end{bmatrix}
\begin{bmatrix}
\yb_1\\
\yb_2\\
\vdots\\
\yb_a
\end{bmatrix}=
\begin{bmatrix}
\Cb\yb_1\\
\Cb\yb_2\\
\vdots\\
\Cb\yb_a
\end{bmatrix}
\end{equation*} 
where $\yb=[\yb_1^T~\yb_2^T~\dots~\yb_a^T]^T$ with $a=\frac{n}{K}=\frac{m}{t}$ and each $\yb_i \in \Fb_q^K,~\forall i \in [a]$. Since $0 < \wtm(\yb) \leq 2\ei$, at the most $2\ei$ vectors among $\yb_1,\yb_2,\dots,\yb_a$ are non-zero. Hence at the most $2\ei$ vectors among $\Cb\yb_1,\Cb\yb_2,\dots,\Cb\yb_a$ are non-zero. Denote $\zb=[\zb_1^T~\zb_2^T~\dots~\zb_a^T]^T$ where each $\zb_i=\Cb\yb_i \in \Fb_q^t,~\forall i \in [a]$. Therefore we have that at the most $2\ei$ vectors among $\zb_1,\zb_2,\dots,\zb_a$ are non-zero. Now for any $\zb \in \ical_{\text{FU}}(\Ab^S,\ei)$ we write
\begin{align} 
&\Sb\zb \neq \pmb{0}\\ 
\Rightarrow &\begin{bmatrix}
\Sb_1 & \Sb_2 \cdots \Sb_a
\end{bmatrix} [\zb_1^T~\zb_2^T~\cdots~\zb_a^T]^T \neq \pmb{0}\\
\Rightarrow ~&\Sb_1\zb_1+\Sb_2\zb_2+\dots+\Sb_a\zb_a \neq \pmb{0}. \label{sz}
\end{align} 
where $\Sb_i \in \Fb_q^{l \times t},~i \in [a]$ is the sub-matrix of $\Sb$ containing $(i-1)t+1^{\thm}$ to $it^{\thm}$ columns of $\Sb$. 
\vspace{2mm}\\
\textbf{I. Case-1, $t \geq 1$, $\ei=1$:}
To satisfy the condition given in (\ref{sz}) for $\ei=1$ it is sufficient that the columns of any two or fewer sub-matrices among $\Sb_1,\Sb_2,\dots,\Sb_a$ form linearly independent set. Hence the columns of each sub-matrix $\Sb_i,~i \in [a]$ are linearly independent. Let $\scal_i$ be the $t$-dimensional subspace of $\Fb_q^l$ generated by the columns of $\Sb_i$ over $\Fb_q$. Now to satisfy the linear independence property of the columns of two or fewer sub-matrices among $\Sb_1,\Sb_2,\dots,\Sb_a$, it is sufficient to have $\scal_i \cap \scal_j =\{\pmb{0}\}$ for any $i,j \in [a]$ and $i \neq j$. Our code construction for an $(\Ab^S,1)$ function update problem where $t \geq 1$ is based on subspace codes.

\begin{const} \label{eps_1}
Our aim is to construct a matrix $\Sb =[\Sb_1~\Sb_2~\dots~\Sb_a] \in \Fb_q^{l \times m}$ where $\Sb_i \in \Fb_q^{l \times t},~i \in [a]$ is the sub-matrix of $\Sb$ containing $(i-1)t+1^{\thm}$ to $it^{\thm}$ columns of $\Sb$ such that the subspaces generated by the columns any two sub-matrices $\Sb_i$ and $\Sb_j$ for $i \neq j,~i,j \in [a]$ are trivially intersecting. Note that for any $i \neq j,~i,j \in [a]$ the subspaces $\scal_i$ and $\scal_j$ generated by the columns of $\Sb_i$ and $\Sb_j$ respectively are $t$ dimensional subspace of $\Fb_q^l$ and satisfies $\scal_i \cap \scal_j=\{\pmb{0}\}$. Hence to construct such $\Sb$ matrix we utilize  pairwise trivially intersecting $t$-dimensional subspaces $\scal_1,\scal_2,\dots,\scal_a$ of $\Fb_q^l$. From the literature on subspace codes \cite{EW_2018,ES_2013,KK_2008}, we know that if $l \geq 2t$ then there exist at least $q^{l-t}$ pairwise trivially intersecting $t$-dimensional subspaces in $\Fb_q^l$. Hence if $q \geq a^{\frac{1}{l-t}}$ and provided $l \geq 2t$ it is possible to find pairwise trivially intersecting $t$-dimensional subspaces $\scal_1,\scal_2,\dots,\scal_a$ of $\Fb_q^l$. Now to construct $\Sb =[\Sb_1~\Sb_2~\dots~\Sb_a]$ we choose a basis of $i^{\thm}$ subspace $\scal_i,~i \in [a]$ which contains $t$ vectors over $\Fb_q^l$ and these $t$ linearly independent vectors form the columns of the sub-matrix $\Sb_i$. After constructing such $\Sb$ matrix, we set $\Hb=\Sb\Ab^S$ which is a valid encoder matrix for the $(\Ab^S,1)$ function update problem with $t \geq 1$. Using this code construction we achieve codelength $l \geq 2t$ for $(\Ab^S,1)$ function update problem if $q \geq a^{\frac{1}{l-t}}$.
\end{const} 
\begin{example}
Consider an $(\Ab^S,\ei)$ function update problem over $\Fb_2$ with $\ei=1$ where $\Ab^S \in \Fb_2^{9 \times 12}$ is given by
\begin{equation*} 
\Ab^S=
\begin{bmatrix}
\Cb & \pmb{0} & \pmb{0}\\
\pmb{0} & \Cb & \pmb{0}\\
\pmb{0} & \pmb{0} & \Cb
\end{bmatrix}
\end{equation*}
where $\Cb \in \Fb_2^{3 \times 4}$ is given by
\begin{equation*} 
\Cb=
\begin{bmatrix}
1 & 0 & 0 & 1\\
0 & 1 & 0 & 1\\
0 & 0 & 1 & 1
\end{bmatrix}.
\end{equation*} 
Now our aim is to construct a matrix $\Sb=[\Sb_1~\Sb_2~\Sb_3] \in \Fb_q^{l \times 9}$ such that the subspaces $\scal_1,~\scal_2,~\scal_3$ generated by the columns of $\Sb_1,~\Sb_2$ and $\Sb_3$ respectively are pairwise trivially intersecting. From our construction we have that it is possible to find $3$ pairwise trivially intersecting $3$-dimensional subspaces of $\Fb_q^l$ if $q \geq 3^{\frac{1}{l-3}}$ and provided $l \geq 6$. If we let $l=6$ then $q \geq 3^{\frac{1}{3}}$ i.e., $q \geq 2$. Hence over $\Fb_2$ it is possible to construct a $6 \times 9$ matrix $\Sb$ such that $\Hb=\Sb\Ab^S$ is a valid encoder matrix for the $(\Ab^S,1)$ function update problem. One possible choice of $3$ pairwise trivially intersecting $3$-dimensional subspaces of $\Fb_q^6$ is $\scal_1=\text{span}\{(1,0,0,0,0,0),(0,1,0,0,0,0),(0,0,1,0,0,0)\}$, $\scal_2=\text{span}\{(0,0,0,1,0,0),(0,0,0,0,1,0),\\(0,0,0,0,0,1)\}$ and $\scal_3=\text{span}\{(1,0,0,1,0,0),(0,1,0,0,1,0),(0,0,1,0,0,1)\}$. Hence the matrix \mbox{$\Sb \in \Fb_2^{6 \times 9}$} is given by 
\begin{equation*} 
\Sb=
\begin{bmatrix}
1 & 0 & 0 & 0 & 0 & 0 & 1 & 0 & 0\\
0 & 1 & 0 & 0 & 0 & 0 & 0 & 1 & 0\\
0 & 0 & 1 & 0 & 0 & 0 & 0 & 0 & 1\\
0 & 0 & 0 & 1 & 0 & 0 & 1 & 0 & 0\\
0 & 0 & 0 & 0 & 1 & 0 & 0 & 1 & 0\\
0 & 0 & 0 & 0 & 0 & 1 & 0 & 0 & 1
\end{bmatrix}.
\end{equation*}
\end{example}               
\vspace{2mm}
II. \textbf{Case-2, $t \geq 1$, $\ei \geq 1$:} Here we provide a linear code construction for the family of $(\Ab^S,\ei)$ function update problem  where $\ei \geq 1$ and $\Ab^S$ is given in (\ref{a_struct}) with $t\geq 1$. To satisfy the condition given in (\ref{sz}) for $\ei \geq 1$ it is sufficient that the columns of any $2\ei$ or fewer sub-matrices among $\Sb_1,\Sb_2,\dots,\Sb_a$ form a linearly independent set.
\begin{const} \label{general}
Our code construction uses a linear code over $\Fb_{q^t}$ of maximum possible dimension with block length $a$ and minimum distance $d_{\mathrm{min}} \geq 2\ei+1$. Let $\hat{\Sb} \in \Fb_{q^t}^{\hat{l} \times a}$ be a parity check matrix of such linear code with $\hat{l}=a-k_{q^t}(a,2\ei+1)$ where $k_{q^t}(a,2\ei+1)$ denotes the maximum dimension of a linear code over $\Fb_{q^t}$ with block length $a$ and minimum distance $d_{\mathrm{min}} \geq 2\ei+1$. Note that any $2\ei$ columns of $\hat{\Sb}$ are linearly independent over $\Fb_{q^t}$. Let $\alpha$ be a primitive element of $\Fb_{q^t}$ and $p(x)=p_0+p_1x+p_2x^2+\dots+p_{t-1}x^{t-1}+x^t$ be the primitive polynomial corresponding to $\alpha$ where each $p_j \in \Fb_q$ for all $j \in \{0,1,\dots,t-1\}$. The corresponding companion matrix is given by 
\begin{equation*}
\Mb=
\begin{bmatrix}
0 & 0 & \dots & 0 & -p_0\\
1 & 0 & \dots & 0 & -p_1\\
0 & 1 & \dots & 0 & -p_2\\
\vdots & \vdots & \ddots & \vdots & \vdots\\
0 & 0 & \dots & 1 & -p_{t-1}
\end{bmatrix}.
\end{equation*}
Now we define a matrix $\Sb=[\Sb_1~\Sb_2~\dots~\Sb_a] \in \Fb_q^{\hat{l}t \times at}$ where for each $j \in [a],~\Sb_j \in \Fb_q^{\hat{l}t \times t}$ is given by $\Sb_j=[\Sb_{1,j}^T~\Sb_{2,j}^T~\dots~\Sb_{\hat{l},j}^T~]^T$. Now for each $i \in [\hat{l}]$ and $j \in [a]$, $\Sb_{i,j} \in \Fb_q^{t \times t}$ is given by 
\begin{equation} \label{S_S'}
\Sb_{i,j}= \left\{ \begin{array}{ccl}
\pmb{0}_{t \times t} & \mbox{if} & \hat{s}_{i,j}=0 \\ 
\Ib_{t \times t} & \mbox{if} & \hat{s}_{i,j}=1 \\ 
\Mb^k & \mbox{if} & \hat{s}_{i,j}=\alpha^k,~k \in \{1,2\dots,q^t-2\}
\end{array} \right.
\end{equation}
where $\hat{s}_{i,j}$ is the $(i,j)^\thm$ entry of $\hat{\Sb}$. Since any $2\ei$ or fewer columns of $\hat{\Sb}$ are linearly independent then using Theorem~3 in \cite{EW_2018} we have that the columns of any $2\ei$ or fewer block matrices among $\Sb_1,\Sb_2,\dots,\Sb_a$ are linearly independent. Hence the matrix $\Hb=\Sb\Ab^S$ is a valid encoder matrix over $\Fb_q$ with codelength $l=\hat{l}t=t(a-k_{q^t}(a,2\ei+1))$. Since any $2\ei$ or fewer columns of $\hat{\Sb}$ are linearly independent we have $\hat{l} \geq 2\ei$ and hence $l \geq 2\ei t$ with equality if and only if $\hat{\Sb}$ is a parity check matrix of an $[a,a-2\ei,2\ei+1]$ MDS code over $\Fb_{q^t}$. Such an MDS code is guaranteed to exist if $q^t \geq a$. Hence using this code construction we achieve codelength $l=2\ei t$ if $q \geq a^{\frac{1}{t}}$.
\end{const}
\begin{example}
Consider an $(\Ab^S,\ei)$ function update problem over $\Fb_2$ with $\ei=2$ where $\Ab^S \in \Fb_2^{15 \times 20}$ is given by
\begin{equation*} 
\Ab^S=
\begin{bmatrix}
\Cb & \pmb{0} & \pmb{0} & \pmb{0} & \pmb{0}\\
\pmb{0} & \Cb & \pmb{0} & \pmb{0} & \pmb{0}\\
\pmb{0} & \pmb{0} & \Cb & \pmb{0} & \pmb{0}\\
\pmb{0} & \pmb{0} & \pmb{0} & \Cb & \pmb{0}\\
\pmb{0} & \pmb{0} & \pmb{0} & \pmb{0} & \Cb
\end{bmatrix}
\end{equation*}
where $\Cb \in \Fb_2^{3 \times 4}$ is given by
\begin{equation*} 
\Cb=
\begin{bmatrix}
1 & 0 & 0 & 1\\
0 & 1 & 0 & 1\\
0 & 0 & 1 & 1
\end{bmatrix}.
\end{equation*} 
Note that $x^3+x+1$ is a primitive polynomial corresponding to $\Fb_8$ and companion matrix corresponding to the primitive polynomial $x^3+x+1=0$ is given by
\begin{equation*} 
\Mb=
\begin{bmatrix}
0 & 0 & 1\\
1 & 0 & 1\\
0 & 1 & 0 
\end{bmatrix}.
\end{equation*}
Now we set $\hat{\Sb}$ as a parity check matrix of a $[5,1,5]$ MDS code over $\Fb_8$ which is repetition code over $\Fb_8$. Hence $\hat{\Sb} \in \Fb_8^{4 \times 5}$ is given by 
\begin{equation*} 
\hat{\Sb}=
\begin{bmatrix}
1 & 0 & 0 & 0 & 1\\
0 & 1 & 0 & 0 & 1\\
0 & 0 & 1 & 0 & 1\\
0 & 0 & 0 & 1 & 1
\end{bmatrix}.
\end{equation*}
Now we obtain the matrix $\Sb \in \Fb_2^{12 \times 15}$ from $\hat{\Sb}$ using (\ref{S_S'}) as 
\begin{equation*} 
\Sb=
\begin{bmatrix}
\Ib_{3 \times 3} & 0 & 0 & 0 & \Ib_{3 \times 3}\\
0 & \Ib_{3 \times 3} & 0 & 0 & \Ib_{3 \times 3}\\
0 & 0 & \Ib_{3 \times 3} & 0 & \Ib_{3 \times 3}\\
0 & 0 & 0 & \Ib_{3 \times 3} & \Ib_{3 \times 3}
\end{bmatrix}.
\end{equation*}
Now we obtain a valid encoder matrix $\Hb=\Sb\Ab^S$ with codelength $12$ over $\Fb_2$.
\end{example}
\vspace{2mm}
\subsubsection{Comparison with the code in Remark~4 of \cite{PM_2018}} \label{compare}
Let us first briefly describe about the system model given in Remark~4 in \cite{PM_2018} using our notations. In Remark~4 of \cite{PM_2018} the authors considered transmission of $t$ updated information symbol vectors $\xb_1+\eb_1,\xb_2+\eb_2,\dots,\xb_t+\eb_t$, $\xb_i+\eb_i \in \Fb_q^K,~\forall i \in [t]$. The receiver knows coded version of each information symbol vector denoted by $\Cb\xb_1,\Cb\xb_2,\dots,\Cb\xb_t$ where $\Cb \in \Fb_q^{t \times K}$ and $\Cb\xb_i \in \Fb_q^{t},~\forall i \in [t]$ and demands updated version of the coded demands i.e., $\Cb(\xb_1+\eb_1),\Cb(\xb_2+\eb_2),\dots,\Cb(\xb_t+\eb_t)$. We can view this problem as an $(\Ab^S,\ei)$ function update problem where $\Ab^S \in \Fb_q^{t^2 \times tK}$ takes the form given in (\ref{a_struct}) with the number of repetitions of the matrix $\Cb$ along the block diagonal entries of $\Ab^S$ being equal to $t$. We denote the information symbol vector as $\xb = [\xb_1^T~\xb_2^T~\dots~\xb_t^T]^T \in \Fb_q^{tK}$ and the update vector as $\eb=[\eb_1^T~\eb_2^T~\dots~\eb_t^T]^T \in \Fb_q^{tK}$ with $\wtm(\eb) \leq \ei$. The authors of \cite{PM_2018} provide a valid code construction with codelength $2t\ei$ based on an MRSC using the Construction 1 in \cite{PM_2018}. This construction from \cite{PM_2018} is valid over any field $\Fb_q$. 

To construct a valid code for the above function update problem we choose $\hat{\Sb}$ as a parity check matrix of a $[t,t-2\ei,2\ei+1]$ MDS code over $\Fb_{q^t}$ and such a code exists if $q^t \geq t$. Then we construct the matrix $\Sb \in \Fb_q^{2t\ei \times t^2}$ from $\hat{\Sb}$ using (\ref{S_S'}). Hence if $q \geq t^{1/t}$ we construct a valid code with codelength $2t\ei$ for the $(\Ab^S,\ei)$ function update problem. Note that for any positive integer $t$, $t^{1/t} < 2$. Hence over any finite field $\Fb_q$ our construction yields a valid encoder matrix with codelength $2t\ei$ for the $(\Ab^S,\ei)$ function update problem described above.     
\vspace{2mm}
\subsubsection{Comparison of Code Construction~1 and Code Construction~2 for $(\Ab^S,1)$ function update problem with $t \geq 1$} \label{code_compare}
In this sub-section we consider the Code Construction~\ref{general} for the special case of $\ei=1$ and then compare the performance with the performance of the Code Construction~\ref{eps_1}. Consider an $(\Ab^S,1)$ function update problem where $\Ab^S$ is of the form given in (\ref{a_struct}). To obtain a valid code for the $(\Ab^S,1)$ function update problem using the Code Construction~\ref{general}, we use a linear code over $\Fb_{q^t}$ of maximum possible dimension  with blocklength $a$ and minimum distance $d_{\text{min}} \geq 3$. Let $\hat{\Sb} \in \Fb_{q^t}^{\hat{l} \times a}$ be a parity of such linear code with $\hat{l}=a-k_{q^t}(a,3)$ where $k_{q^t}(a,3)$ denotes the maximum possible dimension of a linear code over $\Fb_{q^t}$ with blocklength $a$ and minimum distance is at least $3$. We construct a matrix $\Sb \in \Fb_q^{\hat{l} t \times at}$ from $\hat{\Sb}$ using (\ref{S_S'}) and obtain a valid encoder matrix $\Hb$ with code length $\hat{l}t$ by multiplying $\Sb$ with $\Ab^S$. Note that any two or fewer columns of $\hat{\Sb}$ are linearly independent. Hence the subspace generated by each column of $\hat{\Sb}$ are pairwise trivially intersecting. Therefore to construct such a matrix $\hat{\Sb}$ it is necessary and sufficient that the number of trivially intersecting $1$-dimensional subspaces of space $\Fb_{q^t}^{\hat{l}}$ is at least $a$. From \cite{OSC_2005} we know that the space $\Fb_{q^t}^{\hat{l}}$ contains exactly $(q^{\hat{l}t}-1)/(q^t-1)$ trivially intersecting $1$-dimensional  subspaces. Hence to construct a matrix $\Sb$ it is necessary and sufficient that 
\begin{equation*}
\frac{q^{\hat{l}t}-1}{q^t-1} \geq a.
\end{equation*}
Now using the fact $(q^{\hat{l}t}-1)/(q^t-1) \geq q^{\hat{l}t}/q^t$ (since $\hat{l}t \geq t$) we observe that  $q^{\hat{l}t}/q^t \geq a$ i.e., $q \geq a^{\frac{1}{t(\hat{l}-1)}}$ is a sufficient condition for such an encoder matrix to exist. Hence applying the Code Construction~\ref{general} for an $(\Ab^S,1)$ function update problem over $\Fb_q$ we achieve codelength $l=t(a-k_{q^t}(a,3))$ if the field size $q \geq a^{\frac{1}{t(\hat{l}-1)}}$. Hence if $q \geq a^{1/t}$ we achieve codelength $l=2t$ using the Code Construction~\ref{general} by choosing a parity check matrix of an $[a,a-2,3]$ MDS code over $\Fb_{q^t}$ and such a MDS code exists over $\Fb_{q^t}$ since $q^t \geq a$. Note that we also achieve codelength $l=2t$ for $(\Ab^S,1)$ function update problem using the Code Construction~\ref{eps_1} if $q \geq a^{1/t}$. Note that in Code Construction~\ref{general}, the achieved codelength $l=\hat{l}t$ is always an integer multiple of $t$. But applying the Code Construction~\ref{eps_1} for $(\Ab^S,1)$ function update problem we can achieve any codelength $l \geq 2t$ provided the field size $q \geq a^{1/l-t}$. Hence for $(\Ab^S,1)$ function update problem the Code Construction~\ref{general} becomes a special case of the Code Construction~\ref{eps_1}. This also inspires us to study the Code Construction~\ref{eps_1} separately for $(\Ab^S,1)$ function update problem.   

\section{Equivalence with a Functional Index Coding problem} \label{FU_FIC}
In this section we discuss a variation of the classical index coding problem where each user demands a coded version of the information symbols present at the transmitter and already knows a subset of the (uncoded) information symbols as side information. This is a special case of the Generalized Index Coding problem~\cite{DSS_2014,LDH_2015} and the Functional Index Coding problem~\cite{GR_2016}. The authors of \cite{DSS_2014,LDH_2015} generalized the classical index coding problem where each receiver knows some linearly coded information symbols as side-information and demands some linearly coded information symbols. Additionally the authors of \cite{DSS_2014} assume that the information symbols present in the transmitter are also linearly coded information symbols. In \cite{GR_2016}, authors generalized the index coding problem, where the side-information as well as demanded messages can be arbitrary functions of information symbols, called \textit{functional index coding} problem. Here we consider a special case of generalized index coding problem and functional index coding problem and then we introduce the relation between function update problem and this family of functional index coding problems.
\vspace{1mm}
\subsection{Functional Index Coding with Coded Demand and Uncoded Side Information} \label{preli} Consider a broadcast network scenario with single transmitter and $\hat{K}$ receivers $u_1,u_2,\dots,u_{\hat{K}}$. The transmitter has a vector of $n$ information symbols $\xb=(x_1,x_2,\dots,x_n) \in \Fb_q^{n}$. Each receiver knows a subset of the information symbols as side-information. Let $\xb_{\xcal_i}$ be the side-information vector of $i^{\mathrm{th}}$ receiver $u_i$ where $\xcal_i \subseteq [n],~i \in [\hat{K}]$. Each receiver demands a coded version of the information symbols vector $\xb$. Let $\Ab_i\xb$ be the coded demand of $i^{\thm}$ receiver $u_i$ where $\Ab_i \in \Fb_q^{m \times n}$ with rank$(\Ab_i)=m$. Upon denoting $\acal=(\Ab_1, \Ab_2, \dots, \Ab_{\hat{K}})$ and $\xcal=(\xcal_1,\xcal_1,\dots,\xcal_{\hat{K}})$ we describe the problem instance as $(\hat{K},n,\xcal,\acal)$ \textit{functional index coding problem}. A valid \emph{encoding function} $\mathfrak{E}_{\cd}$ over $\Fb_q$ for an $(\hat{K},n,\xcal,\acal)$ functional index coding problem is  
\begin{equation*}
\mathfrak{E}_{\cd}:\Fb_q^n \rightarrow \Fb_q^l
\end{equation*}
such that for each receiver $u_i,$ \mbox{$i \in [\hat{K}]$} there exists a \emph{decoding function}
\mbox{$\mathfrak{D}_{i,\cd}:\Fb_q^l \times \Fb_q^{|\xcal_i|} \rightarrow \Fb_q^m$}
satisfying the following property:
$\mathfrak{D}_{i,\cd}(\mathfrak{E}_{\cd}(\xb),\xb_{\xcal_i})=\Ab_i \xb$
for every \mbox{$\xb \in \Fb_q^n$}.

The design objective is to design a tuple $(\mathfrak{E}_{\cd},\mathfrak{D}_{1,\cd},\mathfrak{D}_{2,\cd},\dots, \mathfrak{D}_{\hat{K},\cd})$ of encoding and decoding functions that minimizes the codelength $l$ and determine the \emph{optimal codelength} for the given functional index coding problem which is the minimum codelength among all coding schemes. 

A linear code for an $(\hat{K},n,\xcal,\acal)$ functional index coding problem is defined as a coding scheme where the encoding function \mbox{$\mathfrak{E}_{\cd}:\Fb_q^n \rightarrow \Fb_q^l$} is a linear transformation over $\Fb_q$ described as \mbox{$\mathfrak{E}_{\cd}(\xb)=\Hb\xb$}, where $\Hb \in \Fb_q^{l \times n}$ is the \emph{encoder matrix} for linear functional index code. The minimum codelength among all valid linear coding schemes for the $(\hat{K},n,\xcal,\acal)$ functional index coding problem over the field $\Fb_q$ will be denoted as $l_{q,\optm,\cd}$.

Now we derive a design criterion for a matrix $\Hb$ to be a valid encoder matrix for $(\hat{K},n,\xcal,\acal)$ functional index coding problem. We define the set $\ical_{\cd}(\hat{K},n,\xcal,\acal)$, or equivalently $\ical_{\cd}$, of vectors $\yb$ of length $n$ such that $\yb_{\xcal_i}={\bf{0}} \in \Fb_q^{|\xcal_i|}$ and $\Ab_i\yb \neq \pmb{0}$ for some choice of $i \in [\hat{K}]$ i.e.,
\begin{equation} \label{i_cd}
\ical_{\cd}(\hat{K},n,\xcal,\acal)=\bigcup\limits_{i=1}^{\hat{K}}\{\yb \in \Fb_q^n~|~\yb_{\xcal_i}={\bf{0}}~\text{and}~\Ab_i\yb \neq \pmb{0}\}.
\end{equation}
\begin{theorem} \label{thm5}
The matrix $\Hb \in \Fb_q^{l \times n}$ is a valid encoder matrix for the $(\hat{K},n,\xcal,\acal)$ functional index coding problem if and only if 
\begin{equation*}
\Hb\yb \neq \pmb{0},~~~\forall \yb \in \ical_{\cd}.
\end{equation*}
\end{theorem}
\begin{proof}
A matrix $\Hb \in \Fb_q^{l \times n}$ is a valid encoder matrix for the $(\hat{K},n,\xcal,\acal)$ functional index coding problem if and only if at each receiver $u_i,~i \in [\hat{K}]$, $\Ab_i\xb$ can be uniquely determined from the received codeword $\Hb\xb$ and the side information $\xb_{\xcal_i}$. Hence for two distinct pair of the information symbol vectors $(\xb,\xb')$ such that the side-information symbol vectors available at the $i^{\text{th}}$ receiver are identical i.e., $\xb_{\xcal_i}=\xb'_{\xcal_i}$ but demanded coded information symbol vectors are distinct i.e., $\Ab_i\xb \neq \Ab_i\xb'$ then the transmitted codeword $\Hb\xb$ must be distinct from $\Hb\xb'$ to distinguish two different demanded coded information symbol vectors. Equivalently, the condition $\Hb\xb \neq \Hb\xb'$ should hold for every pair $\xb,\xb' \in \Fb_q^n$ such that $\Ab_i\xb \neq \Ab_i\xb'$ and $\xb_{\xcal_i}=\xb'_{\xcal_i}$ for some $i \in [K]$. Therefore $\Hb$ is a valid encoder matrix if and only if
\begin{equation*}
\Hb(\xb-\xb') \neq \pmb{0}
\end{equation*}
for all $\xb,\xb' \in \Fb_q^n$ such that $\Ab_i\xb \neq \Ab_i\xb'$ and $\xb_{\xcal_i}=\xb'_{\xcal_i}$ for some $i \in [\hat{K}]$. Now denoting $\yb=\xb-\xb'$ we have
\begin{equation*} 
\Hb\yb \neq \pmb{0}
\end{equation*}
for all $\yb \in \Fb_q^n$ such that $\Ab_i\yb \neq \pmb{0}$ and $\yb_{\xcal_i} = \pmb{0}$ for some $i \in [\hat{K}]$. Hence the statement of the theorem follows.
\end{proof}
\vspace{1mm}
\subsection{Construction of a Equivalent Functional Index Coding Problem from a given Function Update problem} \label{FIC_const}
Now we construct an $(\hat{K},n,\xcal,\acal)$ functional index coding problem starting from an $(\Ab,\ei)$ function update problem. The number of receivers $\hat{K}$, the tuple of the side information indices $\xcal$ and the tuple of coded demands $\acal$ are obtained from Algorithm~1.

\begin{algorithm}[h!]
\caption{Construction of an functional index coding problem from a given function update problem}
\SetAlgoLined
\textbf{Input}: $\Ab \in \Fb_q^{m \times n},\ei$ corresponding to an $(\Ab,\ei)$ function update problem\\
\textbf{Output}: $\hat{K},\xcal,\acal$ corresponding to a functional index coding problem\\
   \% \% Iteration: \\
   $j=0$ \\
  \For{each $Q \subseteq [n]=\{1,2,\dots,n\}$ with $|Q|=\min(2\ei,n)$}{
 $j \leftarrow j+1$\\
 $\xcal_j \leftarrow [n] \setminus Q$\\
 $\Ab_j \leftarrow \Ab$
 }
 \eIf{$n > 2\ei$}{
 $\hat{K}=\binom{n}{2\ei}$}
 {$K=n$}
 $\xcal=(\xcal_1,\xcal_2,\dots,\xcal_{\hat{K}})$\\
 $\acal=(\Ab_1,\Ab_2,\dots,\Ab_{\hat{K}})$
\end{algorithm}
Algorithm~1 considers every possible choice of $Q \subseteq [n]$ such that $|Q|=\min(2\ei,n)$ and defines a new user $u_j$ in the functional index coding problem with demand matrix $\Ab_j=\Ab$ and side information $\xcal_j=[n] \setminus Q$.

Now we relate the set $\mathcal{I}(\Ab,\ei)$ defined for the $(\Ab,\ei)$ Function  Update problem and the set $\ical_{\cd}$ defined in (\ref{i_cd}) for the constructed $(\hat{K},n,\xcal,\acal)$ functional index coding problem.

\begin{theorem}  \label{i_fu_i_cd}
For any given $(\Ab,\ei)$ function update problem and its corresponding $(\hat{K},n,\xcal,\acal)$ functional index coding problem, $\mathcal{I}(\Ab,\ei)=\ical_{\cd}(\hat{K},n,\xcal,\acal)$.   
\end{theorem}
\begin{proof}
To show that $\mathcal{I}(\Ab,\ei)=\ical_{\cd}(\hat{K},n,\xcal,\acal)$, we will show that $\mathcal{I}(\Ab,\ei) \subseteq \ical_{\cd}$ and $\ical_{\cd} \subseteq \mathcal{I}(\Ab,\ei)$.

\vspace{2mm}
\textit{Proof for} $\mathcal{I}(\Ab,\ei) \subseteq \ical_{\cd}$:
Suppose a vector $\yb \in \mathcal{I}(\Ab,\ei)$. Then from (\ref{ical_fu}), we have $\Ab\yb \neq \pmb{0}$ and $0 < \wtm(\yb) \leq 2\ei$. Hence there exists a $Q \subseteq [n]$ such that $|Q|=\min(2\ei,n)$ and $\yb_{[n] \setminus Q}= \mathbf{0}$. Now using the construction procedure described in Algorithm~1 we see that there exists a user $u_j$ in the constructed functional index coding problem such that $\xcal_j=[n] \setminus Q$ and $\Ab_j=\Ab$. The vector $\yb$ satisfies $\yb_{\xcal_j}=\pmb{0}$ and $\Ab_j \yb \neq \pmb{0}$. Hence $\yb \in \ical_{\cd}$.

\vspace{2mm}
\textit{Proof for} $\ical_{\cd} \subseteq \mathcal{I}(\Ab,\ei)$: Suppose a vector $\yb \in \ical_{\cd}$. Then there exists at least one user $j \in [\hat{K}]$ such that $\Ab_j\yb \neq \pmb{0}$ and $\yb_{\xcal_j}=\mathbf{0}$. Since $\Ab_j\yb \neq \pmb{0}$ we have $\yb \neq \pmb{0}$. From Algorithm~1 we see that for any $j \in [\hat{K}]$, $|\xcal_j|= n-\min(2\ei,n)$. Note that $\wtm(\yb)=\wtm(\yb_{\xcal_j}) + \wtm(\yb_{[n]\setminus \xcal_j}) \leq 2\ei$. Again from the construction we have $\Ab_i=\Ab,~\forall j \in [K]$. Therefore $\Ab\yb \neq \pmb{0}$. Hence $\yb \in \mathcal{I}(\Ab,\ei)$.

Hence the theorem holds. 
\end{proof}

Now we relate the problem of constructing linear codes for function update problem to the problem of designing linear coding scheme for the corresponding functional index coding problem.
\begin{theorem}  \label{fu_cd_H}
A matrix $\Hb\in \Fb_q^{l \times n}$ such that $\Hb=\Sb\Ab$ for some matrix $\Sb \in \Fb_q^{l \times m}$ is a valid encoder matrix for the $(\Ab,\ei)$ function update problem if and only if $\Hb$ is a valid encoder matrix for the $(\hat{K},n,\xcal,\acal)$ functional index coding problem.  
\end{theorem}
\begin{proof}
From Theorem~\ref{thm2} we know that $\Hb$ is a valid encoder matrix for the $(\Ab,\ei)$ function update problem if and only if it satisfies
\begin{equation*}
\Hb\yb \neq {\bf{0}},\quad \forall \yb \in \mathcal{I}(\Ab,\ei).
\end{equation*}
Now from Theorem \ref{i_fu_i_cd} we have $\mathcal{I}(\Ab,\ei) = \ical_{\cd}(\hat{K},n,\xcal,\acal)$.
Therefore using Theorem~\ref{thm5} we conclude that $\Hb$ is a valid encoder matrix for the $(\hat{K},n,\xcal,\acal)$ functional index coding problem if and only if $\Hb$ is a valid encoder matrix for the $(\Ab,\ei)$ function update problem.
\end{proof}

\section*{Acknowledgment}
The authors thank Dr V.\ Lalitha for discussions regarding the topic of this paper.




\end{document}